\documentclass[a4paper,11pt,english]{amsart}
\usepackage[utf8]{inputenc}             % can use keyboard characters for accented letters
\usepackage[left = 25mm,top = 30mm,bottom = 25mm,right = 25mm]{geometry}
\usepackage{amsmath,amssymb,amsfonts,dsfont,amsthm}
\usepackage{lmodern}
\usepackage{cmap}                       % make the PDF file "searchable and copyable"
\usepackage{babel}
\usepackage[pdfdisplaydoctitle = true,
      colorlinks = true,
      urlcolor = blue,
      citecolor = blue,
      linkcolor = blue,
      pdfstartview = FitH,
      pdfpagemode = UseNone,
      bookmarksnumbered = true]{hyperref} % to make a pdf file with hyperlinks and
\usepackage{hyperxmp}
\usepackage{float}

\newtheorem{thm}{Theorem}[section]
\newtheorem{cor}[thm]{Corollary}
\newtheorem{lem}[thm]{Lemma}
\newtheorem{prop}[thm]{Proposition}
\theoremstyle{definition}
\newtheorem{defn}[thm]{Definition}
\theoremstyle{remark}
\newtheorem{rem}[thm]{Remark}

\numberwithin{equation}{section}
\renewcommand{\vec}{\pmb}

\newcommand{\RR}{\mathbb{R}}                            % Real numbers
\newcommand{\Ela}{\mathbb{E}\mathrm{la}}                % Space of elasticity tensors
\newcommand{\Piez}{\mathbb{P}\mathrm{iez}}              % Space of Piezo-electricity tensors
\newcommand{\Perm}{\mathbb{P}\mathrm{erm}}              % Space of Permittivity tensors
\newcommand{\HH}{\mathbb{H}}                            % Space of harmonic tensors
\newcommand{\TT}{\mathbb{T}}                            % Space of covariant tensors
\newcommand{\Sym}{\mathbb{S}}                           % Space of totally symmetric covariant tensors
\newcommand{\OO}{\mathrm{O}}                            % Orthogonal group
\newcommand{\SO}{\mathrm{SO}}                           % Special orthogonal group
\newcommand{\Id}{\mathrm{Id}}                           % identity
\newcommand{\ee}{\pmb{e}}                               % vector in R^{3}
\newcommand{\ww}{\pmb{w}}                               % vector in R^{3}
\newcommand{\xx}{\pmb{x}}                               % vector in R^{3}

\newcommand{\bnu}{\pmb{\nu}}                            % normal
\newcommand{\btau}{\pmb{\tau}}                          % perpendicular to normal

\newcommand{\bepsilon}{\pmb{\epsilon}}
\newcommand{\bsigma}{\pmb{\sigma}}

\newcommand{\rp}{\mathrm{p}}                            % polynomial in 3 variables
\newcommand{\id}{\mathbf{1}}
\newcommand{\ba}{\mathbf{a}}

\newcommand{\bd}{\mathbf{d}}
\newcommand{\bh}{\mathbf{h}}

\newcommand{\br}{\mathbf{r}}
\newcommand{\bs}{\mathbf{s}}

\newcommand{\bv}{\mathbf{v}}

\newcommand{\by}{\mathbf{y}}
\newcommand{\bz}{\mathbf{z}}

\newcommand{\bA}{\mathbf{A}}
\newcommand{\bB}{\mathbf{B}}
\newcommand{\lc}{\pmb\varepsilon}               % Levi-Civita tensor
\newcommand{\bC}{\mathbf{C}}     		        % elasticity tensor
\newcommand{\bE}{\mathbf{E}}                    % elasticity tensor
\newcommand{\bS}{\mathbf{S}}                    % totally symmetric tensor
\newcommand{\bT}{\mathbf{T}}                    % generic tensor
\newcommand{\bD}{\mathbf{D}}                    % damage tensor
\newcommand{\bF}{\mathbf{F}}                    % fabric tensor
\newcommand{\bP}{\mathbf{P}}

\DeclareMathOperator{\tr}{tr}

\newcommand{\norm}[1]{\left\Vert#1\right\Vert}  % norm
\newcommand{\set}[1]{\left\{#1\right\}}         % set

\newcommand{\otimesbar}{\; \underline{\overline{\otimes}} \;}
\newcommand{\rcont}[1]{\overset{(#1)}{\cdot}} % r-contraction
\hypersetup{
  pdfauthor={M. Olive, B. Desmorat, B. Kolev, R. Desmorat}, % author
  pdftitle={On the determination of plane and axial symmetries in linear Elasticity and Piezo-electricity}, % title
  pdfsubject={}, % subjet
  pdfkeywords={Anisotropy; Elasticity; Piezo-electricity; Harmonic decomposition} % keywords
  pdflang=en, % language : fr, en
  }
\begin{document}

\title[On the determination of plane and axial symmetries]{On the determination of plane and axial symmetries \protect\\ in linear Elasticity and Piezo-electricity}

\author{M. Olive}
\address[M. Olive]{Universit\'{e} Paris-Saclay, ENS Paris-Saclay, CNRS,  LMT - Laboratoire de M\'{e}canique et Technologie, 94235, Cachan, France}
\email{marc.olive@math.cnrs.fr}

\author{B. Desmorat}
\address[B. Desmorat]{Sorbonne Universit\'{e}, CNRS, Institut Jean Le Rond d'Alembert, UMR 7190, 75005 Paris, France \& Universit\'{e} Paris SUD, Orsay, France}
\email{boris.desmorat@sorbonne-universite.fr}

\author{B. Kolev}
\address[B. Kolev]{Universit\'{e} Paris-Saclay, ENS Paris-Saclay, CNRS,  LMT - Laboratoire de M\'{e}canique et Technologie, 94235, Cachan, France}
\email{boris.kolev@math.cnrs.fr}

\author{R. Desmorat}
\address[R. Desmorat]{Universit\'{e} Paris-Saclay, ENS Paris-Saclay, CNRS,  LMT - Laboratoire de M\'{e}canique et Technologie, 94235, Cachan, France}
\email{rodrigue.desmorat@ens-paris-saclay.fr}

\begin{abstract}
  We formulate necessary and sufficient conditions for a unit vector $\bnu$ to generate a plane or axial symmetry of a constitutive tensor. For the elasticity tensor, these conditions consist of two polynomial equations of degree lower than four in the components of $\bnu$. Compared to Cowin--Mehrabadi conditions, this is an improvement, since these equations involve only the normal vector $\bnu$ to the plane symmetry (and no vector perpendicular to $\bnu$). Similar reduced algebraic conditions are obtained for linear piezo-electricity and for totally symmetric tensors up to order $6$.
\end{abstract}

 %\pacs{46.50.+a ; 91.60.-x; 91.60.Ba}

\maketitle

% ----------------------------------------------------------------
\section{Introduction}
\label{sec:introduction}
% ----------------------------------------------------------------

In 1996, Forte and Vianello~\cite{FV1996} properly defined the eight \emph{symmetry classes} of Elasticity. Later, Chadwick et al \cite{CVC2001}, and then Ting \cite{Tin2003}, proved that the determination of these symmetry classes could be recovered by the calculation of the \emph{plane symmetries} of elasticity tensors (indeed, these reflections always generate their symmetry groups). This result is partially due to the fact that an elasticity tensor is invariant under a second-order rotation $\br(\bnu,\pi)$, \textit{i.e.} a rotation of angle $\pi$ around the unit vector $\bnu$, if and only if it is invariant under the plane symmetry $\bs(\bnu)=-\br(\bnu,\pi)$ across the orthogonal plane to $\bnu$. However, it has been pointed out by the same authors that the determination of the symmetry group of an odd-order tensor by the plane symmetry approach does not hold in general. This applies, in particular, to the piezo-electricity tensor and to higher order tensors. In these cases, determination of symmetry axes are necessary.

The natural polynomial equations which characterize vectors $\bnu$ which are normal to a symmetry plane of a given tensor of order $n$ are of degree $2n$ in the components of $\bnu$. For a totally symmetric tensor of order $n$, this algebraic system consists into $(n + 1) (n + 2)/2$ polynomial equations. For an elasticity tensor, it consists into 21 polynomial equations of degree 8 ranked in a fourth-order tensor (a tensor used to define the \emph{pole figures} as introduced in~\cite{FGB1998}). For a piezo-electricity tensor, we get 18 polynomial equations of degree 6 and for a symmetric second-order tensor $\ba$, 6 polynomial equations of degree 4. In this last case, however, it is well known that this algebraic system can be reduced to the simpler equation $(\ba \cdot \bnu)\times \bnu=0$ (where $\bnu$ is an eigenvector of $\ba$). This last condition is polynomial of degree $n$ in $\bnu$, instead of $2n$, and will thus be referred to as a \emph{reduced algebraic condition}. The question that we address in this paper is the following:
\begin{quote}
  \textit{$Q$: Can we find algebraic equations of lower order/degree that determine the plane/axial symmetries of a given constitutive tensor (and more generally of linear constitutive equations)?}
\end{quote}
In other words, how the usual conditions $(\ba \cdot \bnu)\times \bnu=0$ can be extended to higher order tensors?

For an elasticity tensor $\bE$, partial answers have been provided in~\cite{CM1987,Cow1989,Nor1989}. The study of longitudinal waves in anisotropic media~\cite{Sti1965,Kol1966,Fed1968} has indeed allowed Cowin and Mehrabadi to derive reduced necessary and sufficient conditions for the existence of a plane symmetry for $\bE$. These conditions were first expressed in theorem~\ref{thmCM1987}, using the two independent traces of $\bE$, the \emph{dilatation} tensor $\bd$ and the \emph{Voigt} tensor $\bv$:
\begin{equation}\label{eq:dilatation_voigt}
  \bd  = \tr_{12} \bE, \qquad \bv  = \tr_{13}\bE,
\end{equation}
which are symmetric second-order covariants of $\bE$ \cite{OKDD2018b} (in components, $d_{ij}=E_{kkij}$ and $v_{ij}=E_{kikj}$).

\begin{thm}[Cowin--Mehrabadi (1987)]\label{thmCM1987}
  Let $\bnu$ be a unit vector. Then, $\bnu$ is a normal to a symmetry plane of an elasticity tensor $\bE$, if and only if,
  \begin{equation*}
    \begin{cases}
      \left[ \left(\bnu\cdot \bE\cdot \bnu\right)\cdot \bnu\right] \times \bnu=0,
      \\
      \left[ \left(\btau\cdot \bE\cdot \btau\right)\cdot\bnu\right] \times \bnu=0,
      \\
      \left( \bd\cdot \bnu\right) \times \bnu=0,
      \\
      \left( \bv \cdot \bnu\right)\times \bnu=0,
    \end{cases}
  \end{equation*}
  for all unit vectors $\btau \perp \bnu$.
\end{thm}

When non trivially satisfied, condition $\left( \bd \cdot \bnu\right) \times \bnu=\left( \bv \cdot \bnu\right)\times \bnu=0$ provides, as solution, the common eigenvectors of $\bd$ and $\bv$. It was observed in~\cite{Nor1989} that the third and fourth conditions are in fact consequences of the first two ones.

\begin{thm}[Cowin (1989)]\label{thm:Cowin89}
  Let $\bnu$ be a unit vector. Then $\bnu$ is a normal to a symmetry plane of an elasticity tensor $\bE$ if and only if
  \begin{equation*}
    \begin{cases}
      \left[ \left(\bnu\cdot \bE\cdot \bnu\right)\cdot\bnu\right]  \times \bnu=0,
      \\
      \left[ \left(\btau\cdot \bE\cdot \btau\right)\cdot\bnu\right] \times \bnu=0,
    \end{cases}
  \end{equation*}
  for all unit vectors $\btau \perp \bnu$.
\end{thm}

The Cowin--Mehrabadi conditions are indeed \emph{reduced conditions} in the sense that the first one is a polynomial equation of degree $n=4$ in $\bnu$, instead of degree $2n=8$ for the genuine algebraic conditions. Note, however that the second condition is not very constructive as it requires to check all unit vectors $\btau$ perpendicular to $\bnu$. This drawback is also present in the equivalent forms used by Jaric~\cite{Jar1994} and in the generalized Cowin-Mehrabadi theorems~\cite{Tin2003}. In the present paper, we provide a positive answer to question $Q$ for elasticity tensors, piezo-electricity tensors and totaly symmetric tensors of order 3 up to 6. General but more abstract results, for totally symmetric tensors of any order $n$, have been obtained in~\cite{ODKD2019}. Their proof uses deeply the isomorphism between totally symmetric tensors and homogeneous polynomials in three variables. Except for lemma~\ref{lem:Expr_Second_Order}, the proofs provided in the present paper involve only tensorial operations and are probably more accessible to the mechanical community.

The outline of the paper is the following. In~\autoref{sec:tensorial-operations}, we recall general notions about tensor spaces, such as \emph{covariant operations} on totally symmetric tensors and basic concepts on plane and axial symmetries. All these materials are required to write the proofs of our main results. The special case of symmetric second-order tensors is recalled in~\autoref{sec:second-order-tensors}, where these specific tensorial operations are highlighted. Totally symmetric tensors of order 3 to 6 are studied in~\autoref{sec:third-to-six-order-tensors} , the elasticity tensor in~\autoref{sec:elasticity-tensors} and the piezo-electricity tensor, in~\autoref{sec:piezoelectricity-tensors}.

% ----------------------------------------------------------------
\section{Tensors and plane/axial symmetries}
\label{sec:tensorial-operations}
% ----------------------------------------------------------------

In this section, we recall basic material on tensors and their symmetries. We introduce, furthermore, three basic tensorial operations on tensors; some are well-known in the mathematical community, such as the total symmetrization or the contraction and others are not, like the \emph{generalized cross product}. Finally, we formulate conditions, involving these tensorial operations, which characterize existence of plane/axial symmetries of tensors.

An $n$th-order tensor $\bT \in \TT^{n}(\RR^3)$ is an $n$-linear form
\begin{equation*}
  (\xx_{1},\dotsc,\xx_n)\in \RR^{3}\times \dotsc \times \RR^{3}\mapsto \bT(\xx_{1},\dotsc,\xx_n)\in \RR.
\end{equation*}
In any orthonormal basis $(\ee_{1},\ee_{2},\ee_{3})$, its components write $T_{i_{1} i_{2} \dotsc i_n} = \bT(\ee_{i_{1}}, \ee_{i_{2}}, \dotsc,\ee_{i_n} )$, since
\begin{equation*}
  \bT(\xx_{1},\xx_{2},\dotsc,\xx_n):=T_{i_{1} i_{2} \dotsc i_n} (\xx_{1})_{i_{1}} (\xx_{2})_{i_{2}}\dotsc( \xx_n)_{i_n}.
\end{equation*}
The natural action of an orthogonal transformation $g\in \OO(3)$ on $\bT\in \TT^{n}(\RR^3)$ is given by
\begin{equation}\label{eq:Def_Lin_Rep}
  (g\star\bT)(\xx_{1},\dotsc,\xx_n):=\bT(g^{-1}\xx_{1},\dotsc,g^{-1}\xx_n),
\end{equation}
or, in components in any orthonormal basis, as
\begin{equation*}
  (g \star \bT)_{i_{1}\dotsc i_n} = g_{i_{1}j_{1}}\dotsm g_{i_nj_n}T_{j_{1}\dotsc j_n}.
\end{equation*}
In the following, $g$ will be either the plane symmetry
\begin{equation}\label{eq:sdenu}
  \bs(\bnu):=\Id-2 \bnu \otimes \bnu, \qquad \norm{\bnu}=1,
\end{equation}
where $\bnu$ is a unit normal to the considered plane of symmetry and $\Id$ is the identity transformation, or the order-two rotation around the axis $\langle\bnu\rangle$,
\begin{equation*}
  \br(\bnu, \pi):=-\bs(\bnu)\in \SO(3).
\end{equation*}

It follows that a unit vector $\bnu$ is a normal to a symmetry plane of a given tensor $\bT$ if and only if
\begin{equation*}
  \bs(\bnu)\star \bT=\bT,
\end{equation*}
and that $\bnu$ spans a symmetry axis of $\bT$ if and only if
\begin{equation*}
  \br(\bnu, \pi)\star \bT=\bT,
\end{equation*}
both conditions being equivalent for even order tensors.

Let $\bnu$ be a unit vector, $g \in \OO(3)$ such that $g\bnu = \ee_{3}$, then,
\begin{equation*}
  \bs(\bnu) = \bs(g^{-1}\ee_{3}) = g^{-1}\bs(\ee_{3})g,
\end{equation*}
and thus, if $\bT\in \TT^{n}(\RR^3)$ is an $n$-th order tensor, we have
\begin{equation*}
  \bs(\bnu) \star \bT = \left(g^{-1}\bs(\ee_{3})g\right) \star \bT = g^{-1} \star \left( \bs(\ee_{3}) \star (g \star \bT) \right).
\end{equation*}
One has therefore the equivalence
\begin{equation}\label{eq:e3_instead_bnu}
  \bs(\bnu) \star \bT = \bT \iff \bs(\ee_{3})\star(g\star \bT) = g\star \bT,
\end{equation}
which means that $\bs(\bnu)$ is a plane symmetry of $\bT$ if and only if $\bs(\ee_{3})$ is a plane symmetry of $g\star \bT$.

\begin{rem}\label{rem:Sym_and_tensor_Coordinates}
  As it is well known in Elasticity, one can recast in terms of the components $T_{i_{1}\dotsc i_n}$ of $\bT\in \TT^{n}(\RR^3)$
  the conditions that $\bnu=\pmb{e}_{3}$ is a normal to a symmetry plane, or that
  $\langle \bnu \rangle=\langle\ee_{3}\rangle$  is a symmetry
  axis. More precisely, in an orthonormal basis  $(\ee_{1},\ee_{2},\ee_{3})$:
  \begin{enumerate}
    \item $\bs(\ee_{3})\star \bT=\bT$ if and only if
          \begin{equation*}
            T_{i_{1}\dotsc i_n} = 0,
          \end{equation*}
          whenever the occurrences of $3$ in the set $\set{i_{1},\dotsc ,i_n}$ is odd.
    \item For order $n$ odd, then $\br(\ee_{3},\pi)\star \bT=\bT$ if and only if
          \begin{equation*}
            T_{i_{1}\dotsc i_n}=0,
          \end{equation*}
          whenever the occurrences of $3$ in the set $\set{i_{1},\dotsc ,i_n}$ is even.
  \end{enumerate}
\end{rem}

The space $\Sym^{n}(\RR^3)$ of $n$th-order totally symmetric tensors is defined as the space of tensors $\bS \in \TT^{n}(\RR^3)$, such that $ \bS=\bS^{s}$ with
\begin{equation*}
  \bT^{s}(\xx_{1},\dotsc,\xx_{n}) : = \frac{1}{n!}\sum_{\sigma \in \mathfrak{S}_{n}} \bT(\xx_{\sigma(1)},\dotsc,\xx_{\sigma(n)}),
\end{equation*}
$\mathfrak{S}_{n}$ being the permutation group of $n$ elements.

\begin{rem}\label{rem:Total_Sym_Tens_Pol}
  There is a covariant isomorphism
  \begin{equation*}
    \bS\in \Sym^{n}(\RR^3)\mapsto \rp(\xx):=\bS(\xx,\dotsc,\xx),
  \end{equation*}
  which associates to each $n$th-order totally symmetric tensor $\bS$, an homogeneous polynomial $\rp$ of degree $n$ in $\xx=(x,y,z)\in \RR^3$. The inverse mapping is obtained by \emph{polarization} (see~\cite[p. 35]{Gor2017} or~\cite[Section 2.1]{OKDD2018a} for more details). Given an homogeneous polynomial $\rp(\xx)$ in $\xx=(x,y,z)$, one gets
  \begin{equation*}
    \bS(\xx_{1},\dotsc,\xx_{n}) = \frac{1}{n!} \left.\frac{\partial^{n}}{\partial t_{1} \dotsb \partial t_{n}}\right|_{t_{1} = \dotsb = t_{n} = 0}\rp(t_{1}\xx_{1} + \dotsb + t_{n}\xx_{n}).
  \end{equation*}
\end{rem}

Next, we introduce several tensorial operations between totally symmetric tensors.

\begin{defn}[Symmetric tensor product]
  The \emph{symmetric tensor product} between two totally symmetric tensors $\bS^{1} \in \Sym^{p}(\RR^{3})$ and $\bS^{2} \in \Sym^{q}(\RR^{3})$ is defined as
  \begin{equation*}
    \bS^{1} \odot \bS^{2} : = (\bS^{1} \otimes \bS^{2})^{s} \in \Sym^{p + q}(\RR^{3}).
  \end{equation*}
\end{defn}

Given a vector $\ww\in \RR^{3}$, we will write (as in~\cite{Qi2007})
\begin{equation*}
  \ww^{k} : = \underbrace{\ww \otimes \ww \otimes \dotsb \otimes \ww}_\textrm{$k$ times},
\end{equation*}
with components $\ww^{k}_{i_{1} i_{2}\dotsc i_k} = w_{i_{1}}w_{i_{2}}\dotsc w_{i_k}$. Of course,
$\ww \otimes \ww \otimes \dotsb \otimes \ww=\ww \odot \ww \odot \dotsb \odot \ww$.

\begin{defn}[$r$--contraction]\label{def:r-cont}
  The $r$-contraction between an $n$th-order symmetric tensor $\bS \in \Sym^{n}(\RR^3)$ and $\ww^r$ (where $r\leq n$) is defined by
  \begin{equation*}
    (\bS \rcont{r} \ww^r)_{i_{1}i_{2}\dotsc i_{n-r}} = S_{i_{1}i_{2}\dotsc i_{n-r}j_{1}\dotsc j_{r}}w_{j_{1}}\dotsm w_{j_r},
    \qquad \bS \rcont{r} \ww^r \in  \Sym^{n-r},
  \end{equation*}
  where Einstein's convention on repeated indices has been adopted.
\end{defn}

\begin{defn}[Generalized cross product \cite{DADKO2019}]
  The \emph{generalized cross product} between two totally symmetric tensors $\bS^{1} \in \Sym^{p}(\RR^{3})$ and $\bS^{2} \in \Sym^{q}(\RR^{3})$ is defined as
  \begin{equation}\label{eq:cross-product}
    \bS^{1} \times \bS^{2} : =  \left(\bS^{2}\cdot\lc \cdot \bS^{1}\right)^{s} \in \Sym^{p + q -1}(\RR^{3}),
  \end{equation}
  where $\lc$ is the Levi-Civita symbol in $\RR^3$.
\end{defn}

In any direct orthonormal basis $(\ee_i)$, it writes as
\begin{equation}\label{eq:Gen_Cross_Coord}
  (\bS^{1}\times\bS^{2})_{i_{1}\dotsb i_{p+q-1}} : = (\varepsilon_{i_{1}jk}S^{1}_{ji_{2}\dotsb i_{p}}S^{2}_{ki_{p+1} \dotsb i_{p+q-1}})^{s},
  \qquad
  \varepsilon_{ijk} = \det(\ee_i, \ee_j, \ee_j).
\end{equation}

\begin{rem}\label{rem:wkxw}
  Let $\ww$ be a vector and $\ww^{k} = \ww \otimes \ww \otimes \dotsb \otimes\ww$ be the tensor product of $k$ copies of $\ww$, then,
  \begin{equation*}
    \ww^{k} \times \ww=0.
  \end{equation*}
\end{rem}

We come now to some properties which relate tensorial symmetries with these tensorial operations.

\begin{lem}\label{lem:Expr_Second_Order}
  Let $\bS$ be any totally symmetric $n$th-order tensor. Then,
  \begin{equation}\label{eq:gTp0}
    \bs(\bnu)\star \bS=\sum_{k=0}^{n} \binom{n}{k} (-2)^{k}\bnu^{k} \odot \left(\bS\rcont{k} \bnu^{k}\right),
  \end{equation}
  and
  \begin{equation}\label{eq:gTp1}
    \br(\bnu,\pi)\star \bS=(-1)^{n}\bs(\bnu)\star\bS.
  \end{equation}
\end{lem}

\begin{proof}
  Taking account of remark~\ref{rem:Total_Sym_Tens_Pol}, it is sufficient to prove the lemma for the homogeneous polynomial,
  \begin{equation*}
    \rp(\xx) := \left(\bs(\bnu)\star \bS\right)(\xx,\dotsc,\xx) = \bS(\xx-2\langle\bnu,\xx\rangle \bnu,\dotsc, \xx-2\langle\bnu,\xx\rangle \bnu),
  \end{equation*}
  where $\langle \xx ,\bnu \rangle=\xx\cdot \bnu$ is the standard inner product on $\RR^3$, rather than for the totally symmetric tensor $\bs(\bnu)\star \bS$ itself. The tensor $\bS$ being totally symmetric, we get
  \begin{equation*}
    \rp(\xx) = \sum_{k=0}^{n} \binom{n}{k} (-2)^{k} \langle \bnu,\xx\rangle^{k} \bS(\underbrace{\bnu,\dotsc,\bnu}_{k \text{times}},\xx,\dotsc,\xx),
  \end{equation*}
  where
  \begin{equation*}
    \langle \bnu,\xx\rangle^{k} \bS(\bnu,\dotsc,\bnu,\xx,\dotsc,\xx) = \left(\bnu^{k} \odot \Big(\bS\rcont{k}\bnu^{k}\Big) \right)(\xx,\dotsc,\xx).
  \end{equation*}
  This leads to~\eqref{eq:gTp0}, and~\eqref{eq:gTp1} follows from $\br(\bnu,\pi)=-\bs(\bnu)$, which achieves the proof.
\end{proof}

We now use the fact that the $r$--contraction (see definition~\ref{def:r-cont}) is a covariant operation, which means that for any unit vector $\bnu$, and for any $\bS\in \Sym^{n}(\RR^3)$, $g\in \OO(3)$:
\begin{equation*}
  g\star \left( \bS \rcont{k} { \bnu^{ k}}\right)=\left(g\star\bS\right) \rcont{k} (g\bnu)^{k},
\end{equation*}
where $(g\bnu)^{k}=g\bnu\otimes \dotsc \otimes g\bnu$. Applying this property to $g=\bs(\bnu)$ leads to the following result.

\begin{lem}\label{lem:Sym_With_Contraction}
  Let $\bS\in \Sym^{n}(\RR^3)$ and $\bnu$ be a unit vector.
  \begin{enumerate}
    \item If $\bs(\bnu)\star \bS=\bS$  (\textit{i.e.} $\bnu$ defines a plane symmetry of $\bS$), then for any non--zero integer $k$
          \begin{equation*}
            \bs(\bnu)\star  \big(\bS \rcont{k} { \bnu^{k}}\big)=(-1)^{k} \,\bS \rcont{k} { \bnu^{ k}}.
          \end{equation*}
    \item If $\br(\bnu,\pi)\star \bS=\bS$  (\textit{i.e.} $\bnu$ defines an axial symmetry of $\bS$), then then for any non--zero integer $k$
          \begin{equation*}
            \br(\bnu,\pi)\star  \big(\bS \rcont{k} { \bnu^{k}}\big)=\bS \rcont{k} { \bnu^{ k}}.
          \end{equation*}
  \end{enumerate}
\end{lem}

We have furthermore the property that
\begin{equation*}
  \bs(\bnu)\star \bS = \bS \implies \bnu\cdot(\bs(\bnu)\star \bS)=\bnu\cdot \bS.
\end{equation*}
Now, one way to rewrite the expanded formula for $\bnu\cdot(\bs(\bnu)\star \bS)$ in lemma~\ref{lem:Expr_Second_Order} in a more explicit form is to use the following contraction formula
\begin{equation}\label{eq:Contraction_Formula}
  \bnu \cdot \left(\bnu^{ k} \odot \Big(\bS \rcont{k} { \bnu^{ k}}\Big)\right)
  = \frac{1}{n} \left(k\, \bnu^{ k-1} \odot \Big(\bS \rcont{k}{ \bnu^{ k}}\Big)
  + (n-k)\, \bnu^{ k} \odot \Big(\bS \rcont{k+1} { \bnu^{ k+1}}\Big)\right).
\end{equation}
where $\bS$ is a totally symmetric $n$th-order tensor and $k\leq n-1$. A consequence of lemma~\ref{lem:Sym_With_Contraction} is the following.

\begin{cor}\label{cor:sTnumoinsTnu}
  Let $\bS\in \Sym^{n}(\RR^3)$ and $\bnu$ be a unit vector.
  \begin{enumerate}
    \item For any even integer $k=2p\leq n$, if $\bnu$ is a normal to a symmetry plane of $\bS$, then, it is also a normal to a symmetry plane of $\bS\rcont{2p} \bnu^{2p}\in \Sym^{n-2p}(\RR^3)$.
    \item For any odd order totally symmetric tensor $\bS\in \Sym^{2p+1}(\RR^3)$,
          \begin{equation*}
            \bs(\bnu)\star\bS = \bS \implies \bS(\bnu, \dotsc, \bnu)=\bS\rcont{2p+1} \bnu^{2p+1}=0.
          \end{equation*}
    \item  For any even order totally symmetric tensor $\bS\in \Sym^{2p}(\RR^3)$,
          \begin{equation*}
            \bs(\bnu)\star\bS = -\bS \implies \bS(\bnu, \dotsc, \bnu)=\bS\rcont{2p} \bnu^{2p}=0.
          \end{equation*}
  \end{enumerate}
\end{cor}

% ----------------------------------------------------------------
\section{Plane/axial symmetries for second-order symmetric tensors}
\label{sec:second-order-tensors}
% ----------------------------------------------------------------

We will start our investigation by the case of second-order tensors, for which reduced equations for plane and axial symmetries are easy to understand. Let $\ba\in \Sym^2(\RR^3)$ be a second order symmetric tensor. The action of an orthogonal transformation $g\in \OO(3)$ on $\ba$ writes
\begin{equation*}
  g \star \ba=g\ba g^{t}, \qquad (g \star \ba)_{ij}=g_{ik}g_{jl} a_{kl},
\end{equation*}
where $(\cdot)^{t}$ stands for the transpose. Then, by lemma~\ref{lem:Expr_Second_Order}, we get
\begin{equation}\label{eq:gT2}
  \bs(\bnu)\star \ba=\ba-4 \bnu \odot (\ba \cdot \bnu)+4 ( \bnu\cdot \ba \cdot\bnu) \bnu \otimes \bnu,
\end{equation}
where $\bnu\cdot\ba \cdot\bnu=\ba\rcont{2}(\bnu\odot \bnu)$. Since $\ba$ is of even order, note that

\begin{equation*}
  \bs(\bnu)\star \ba=\br(\bnu,\pi)\star \ba.
\end{equation*}
Hence, any axial symmetry is a plane symmetry of $\ba$, and \emph{vice--versa}. This observation is moreover still true for any even order tensor. A direct application of~\eqref{eq:gTp0} leads to
\begin{equation}\label{eq:sa}
  \bs(\bnu)\star\ba=\br(\bnu,\pi)\star \ba=\ba\Longleftrightarrow (\ba\cdot\bnu)\odot \bnu-( \bnu\cdot \ba \cdot\bnu)\bnu\otimes \bnu =0,
\end{equation}
which is a condition for $\bnu$ to be a normal/axis of a plane/axial symmetry of $\ba$. The standard action~\eqref{eq:gT2} on $\ba$ gives thus 6 polynomial equations of degree $4$ in the coordinates of $\bnu=(x,y,z)\in \RR^3$. We shall see now that one can reduce them to 3 polynomial equations of degree $2$ in $\bnu$. Note that a distinction is made between tensors, for which condition \eqref{eq:axnu} applies, and pseudo-tensors, for which condition \eqref{eq:Second_order_flip} applies. This distinction will arise naturally in the harmonic decomposition of piezo-electricity tensor \cite{JCB1978,GW2002}, in~\autoref{sec:piezoelectricity-tensors}.

\begin{prop}\label{prop:Axial_Second_Order}
  Let $\ba\in\Sym^2(\RR^3)$ and $\bnu$ be a unit vector. Then
  \begin{enumerate}
    \item $\bs(\bnu)\star \ba=\ba$, \emph{i.e} $\bs(\bnu)$ is a plane symmetry (and $\br(\bnu,\pi)$ is an axial symmetry) of $\ba$ if and only if
          \begin{equation}\label{eq:axnu}
            (\ba\cdot \bnu)\times \bnu = 0,
          \end{equation}
          where $\times$ is the cross product.
    \item $\bs(\bnu)\star \ba=-\ba$, \emph{i.e} $\bs(\bnu)$ is a plane symmetry of the pseudo-tensor $\ba$, if and only if
          \begin{equation}\label{eq:Second_order_flip}
            \ba-2\,\bnu \odot  (\ba\cdot\bnu) = 0.
          \end{equation}
  \end{enumerate}
\end{prop}

\begin{rem}
  Condition~\eqref{eq:axnu} is well known. However, since its proof highlights the one for higher order tensors, we provide it anyway.
\end{rem}

\begin{proof}
  (1) Assume first that $\bs(\bnu)\star \ba=\ba$. Hence, by corollary~\ref{cor:sTnumoinsTnu} and the definition of $\bs(\bnu)$, we have
  \begin{equation*}
    \bs(\bnu)\star (\ba\cdot \bnu) =-\ba\cdot \bnu=\ba\cdot \bnu-2(\bnu\otimes \bnu)\cdot (\ba\cdot \bnu)=\ba\cdot \bnu-2(\bnu\cdot \ba\cdot \bnu)\bnu,
  \end{equation*}
  and thus
  \begin{equation*}
    \ba\cdot\bnu - ( \bnu\cdot \ba \cdot\bnu) \bnu = 0.
  \end{equation*}
  Taking the cross product of this expression with $\bnu$ leads to \eqref{eq:axnu}. Conversely, if $(\ba\cdot \bnu)\times \bnu = 0$, then $\bnu$ is an eigenvector of $\ba$, but any eigenvector of $\ba$ is normal to a symmetry plane of $\ba$, meaning that $\bs(\bnu)\star \ba=\ba$.

  (2) Assume now that $\bs(\bnu)\star \ba=-\ba$ so that~\eqref{eq:gT2} leads to
  \begin{equation*}
    \ba-2(\ba\cdot\bnu)\odot \bnu-2\left(\bnu\cdot \ba \cdot \bnu)\right)\bnu\otimes \bnu=0.
  \end{equation*}
  Using point (3) of corollary~\ref{cor:sTnumoinsTnu}, we deduce that $\bnu\cdot\ba\cdot \bnu=0$ and thus, we deduce~\eqref{eq:Second_order_flip}. Conversely, suppose that~\eqref{eq:Second_order_flip} holds. Contracting two times both sides of this equation with $\bnu$ and using~\eqref{eq:gT2}, we get
  \begin{equation*}
    \bs(\bnu)\star \ba=\ba-4 \bnu \odot (\ba \cdot \bnu) + 4 (\bnu\cdot \ba \cdot\bnu) \bnu \otimes \bnu=-\ba,
  \end{equation*}
  since $\bnu\cdot\ba\cdot\bnu=0$ (by contraction formula \ref{eq:Contraction_Formula}). This concludes the proof.
\end{proof}

% ----------------------------------------------------------------
\section{The case of totally symmetric tensors of order three to six}
\label{sec:third-to-six-order-tensors}
% ----------------------------------------------------------------

Let $\bS\in \Sym^{n}(\RR^3)$ be a totally symmetric tensor of order $n$. A necessary and sufficient condition for $\bnu$ ($\norm{\bnu}=1$) to be a normal to a symmetry plane of $\bS$ is
\begin{equation*}
  \bs(\bnu)\star \bS=\bS,
\end{equation*}
where $\bs(\bnu)=\Id-2 \bnu \otimes \bnu$. A necessary and sufficient condition for $\bnu$, to be the axis of an axial symmetry for an \emph{odd order} tensor $\bS$ is
\begin{equation*}
  \bs(\bnu)\star \bS=-\br(\bnu, \pi)\star \bS = -\bS,
\end{equation*}
where $\br(\bnu, \pi) = -\bs(\bnu)$. For even order tensors both conditions are equivalent. These equations are polynomial and of degree $2n$ in $\bnu$. They can be ranked into a totally symmetric tensor of order $n$ (leading to $(n+1)(n+2)/2$ scalar equations).

In order to formulate reduced algebraic equations, of lower degree in $\bnu$, determining plane and axial symmetries of $\bS$, we will use tensorial operations defined in~\autoref{sec:tensorial-operations} and, in particular, the generalized cross product~\eqref{eq:cross-product} between two totally symmetric tensors. For higher order tensors, see \cite{ODKD2019}. These reduced equations are obtained using similar ideas as the ones used in the proof of proposition~\ref{prop:Axial_Second_Order} for second-order tensors.

\begin{thm}\label{thm:S3plane}
  Let $\bS$ be a totally symmetric tensor of order 3 and $\bnu$ be a unit vector, then $\bs(\bnu)\star \bS=\bS$ (\textit{i.e.} $\bnu$ defines a plane symmetry of $\bS$) if and only if
  \begin{equation}\label{eq:condT3}
    \bS\cdot \bnu -2 \bnu \odot \left(\bnu\cdot \bS\cdot \bnu\right)=0.
  \end{equation}
\end{thm}

\begin{proof}
  Suppose first that $\bs(\bnu)\star \bS=\bS$. Then, by corollary~\ref{cor:sTnumoinsTnu} we get
  \begin{equation*}
    \bs(\bnu)\star (\bS\cdot \bnu) = -(\bS\cdot \bnu), \quad \text{and} \quad \bS\rcont{3}\bnu^{3}=0.
  \end{equation*}
  Finally, applying~\eqref{eq:gTp1} to the second-order tensor $\bS\cdot \bnu$ leads to
  \begin{equation*}
    \bS\cdot \bnu-2 \bnu \odot (\bS\rcont{2}\bnu^{2})+2 \bnu^{2} \odot  (\bS\rcont{3}\bnu^{3})=\bS\cdot \bnu-2 \bnu \odot (\bS\rcont{2}\bnu^{2}) = 0,
  \end{equation*}
  and hence we deduce~\eqref{eq:condT3}. Conversely, suppose~\eqref{eq:condT3} holds. Contacting twice this equation with $\bnu$ and using contraction formula \eqref{eq:Contraction_Formula} gives $(\bnu\cdot \bS\cdot \bnu)\cdot \bnu=0$. Finally, by~\eqref{eq:gTp0}, we obtain
  \begin{align*}
    \bs(\bnu)\star \bS & =\bS-6 \bnu \odot (\bS\cdot \bnu)+12 \bnu \odot \bnu \odot (\bnu\cdot \bS \cdot \bnu) -8 [( \bnu\cdot \bS \cdot\bnu)\cdot \bnu] \,\bnu \otimes \bnu \otimes \bnu , \\
                       & = \bS-6 \bnu \odot (\bS\cdot \bnu)+12 \bnu \odot \bnu \odot (\bnu\cdot \bS \cdot \bnu),                                                                            \\
                       & =\bS .
  \end{align*}
\end{proof}

The proofs of the remaining results follow the same lines: the reduced equations are obtained using corollary~\ref{cor:sTnumoinsTnu} (sometimes several times) together with formula~\eqref{eq:gTp0} or~\eqref{eq:gTp1} and remark~\ref{rem:wkxw} for $\ww=\bnu$. Besides, choosing $\bnu=\ee_{3}$ in the reduced equation, leads to simpler equations when using the components of $\bS$, so that we can conclude using remark~\ref{rem:Sym_and_tensor_Coordinates}. Next, we formulate these ideas for the axial symmetry of a totally symmetric third-order tensor.

\begin{thm}\label{thm:S3axial}
  Let $\bS$ be a totally symmetric tensor of order 3 and $\bnu$ be a unit vector, then $\br(\bnu,\pi)\star \bS=\bS$ (\emph{i.e} $\bnu$ defines an axis of symmetry of $\bS$) if and only if
  \begin{equation*}
    \Big[\bS -3 \bnu \odot (\bS\cdot \bnu)\Big]\times \bnu=0.
  \end{equation*}
\end{thm}

\begin{proof}
  Suppose first that $\br(\bnu,\pi) \star \bS=-\bs(\bnu) \star \bS=\bS$, so~\eqref{eq:gTp1} leads to
  \begin{equation}\label{eq:Non_Reduced_Sym3}
    \bS-3 \bnu \odot (\bS\cdot \bnu)+6 \bnu \odot \Big( \bnu \odot (\bnu\cdot \bS \cdot \bnu) -4 [( \bnu\cdot \bS \cdot\bnu)\cdot \bnu] \, \bnu \otimes \bnu \Big)=0.
  \end{equation}
  Now by lemma~\ref{lem:Sym_With_Contraction} and~\eqref{eq:gT2}, one has
  \begin{equation*}
    \bs(\bnu) \star (\bS\cdot \bnu)=\bS\cdot \bnu=\bS\cdot \bnu-4 \bnu \odot ((\bS\cdot \bnu) \cdot \bnu)+4 ( \bnu\cdot (\bS\cdot \bnu) \cdot\bnu) \bnu \otimes \bnu ,
  \end{equation*}
  and $\bS$ being totally symmetric, we deduce that
  \begin{equation*}
    \bnu \odot (\bnu\cdot \bS \cdot \bnu)=[( \bnu\cdot \bS \cdot\bnu)\cdot \bnu] \,\bnu \otimes \bnu .
  \end{equation*}
  Hence,~\eqref{eq:Non_Reduced_Sym3} recasts as
  \begin{equation*}
    \bS-3 \bnu \odot (\bS\cdot \bnu)+2 [( \bnu\cdot \bS \cdot\bnu)\cdot \bnu] \,\bnu \otimes \bnu \otimes \bnu =0.
  \end{equation*}
  Finally, by remark~\ref{rem:wkxw}, the fact that $( \bnu\cdot \bS \cdot\bnu)\cdot \bnu\in \RR$, and taking the generalized cross-product with $\bnu$, we get
  \begin{equation*}
    \big[\bS -3 \bnu \odot (\bS\cdot \bnu)\big]\times \bnu=0.
  \end{equation*}
  Conversely, suppose that $\bnu=\ee_{3}$ and that the third-order tensor
  \begin{equation*}
    \bB= \big[\bS -3 \ee_{3} \odot (\bS\cdot \ee_{3})\big]\times \ee_{3}
  \end{equation*}
  vanishes. Then, using~\eqref{eq:Gen_Cross_Coord}, we have
  \begin{align*}
    B_{111} & = S_{112}=0,
            &                           & B_{112}= \frac{1}{3}(2 S_{122}-S_{111})=0,
            &                           & B_{122}= \frac{1}{3}(S_{222}-2 S_{112})=0,
    \\
    B_{133} & = -\frac{1}{3} S_{233}=0,
            &                           & B_{222}= -S_{122}=0,
            &                           & B_{233}= \frac{1}{3} S_{133}=0.
  \end{align*}
  Hence, $S_{i_{1} i_{2} i_{3}}=0$, whenever the occurrence of $3$ in the set $\set{i_{1},i_{2},i_{3}}$ is even, and thus $\ee_{3}$ defines an axis of symmetry of $\bS$ (by remark~\ref{rem:Sym_and_tensor_Coordinates}). We conclude using~\eqref{eq:e3_instead_bnu}.
\end{proof}

For totally symmetric tensors of order four, five and six, we have the following results.

\begin{thm}\label{thm:S4toS6}
  Let $\bnu$ be a unit vector.
  \begin{enumerate}
    \item For any $\bS\in \Sym^{4}(\RR^3)$, then, $\bnu$ defines a plane symmetry --- or an axial symmetry --- of $\bS$ if and only if
          \begin{equation*}
            \Big[\bS\cdot \bnu -3 \bnu \odot \left(\bnu\cdot \bS\cdot \bnu\right)\Big]\times \bnu=0.
          \end{equation*}
    \item For any $\bS\in \Sym^{5}(\RR^3)$, then, $\bnu$ defines a plane symmetry of $\bS$ if and only if
          \begin{equation*}
            \bS\cdot \bnu -4 \bnu \odot \left(\bnu\cdot \bS\cdot \bnu\right)
            + 4 \bnu \odot \bnu \odot
            \big(\left(\bnu\cdot \bS\cdot \bnu\right)\cdot \bnu\big) = 0.
          \end{equation*}
    \item For any $\bS\in \Sym^{5}(\RR^3)$, then, $\bnu$ defines an axial symmetry of $\bS$ if and only if
          \begin{equation*}
            \Big[\bS -5 \bnu \odot (\bS\cdot \bnu)
              + \frac{20}{3} \bnu \odot \bnu \odot
              (\bnu\cdot \bS\cdot \bnu)
              \Big]\times \bnu = 0.
          \end{equation*}
    \item For any $\bS\in \Sym^{6}(\RR^3)$, then, $\bnu$ defines a plane symmetry --- or an axial symmetry --- of $\bS$ if and only if
          \begin{equation*}
            \big[\bS\cdot \bnu -5 \bnu \odot \left(\bnu\cdot \bS\cdot \bnu\right)
              + \frac{20}{3} \bnu \odot \bnu \odot
              \big(\left(\bnu\cdot \bS\cdot \bnu\right)\cdot \bnu\big)
              \big]\times \bnu = 0.
          \end{equation*}
  \end{enumerate}
\end{thm}

\begin{proof}[Sketch of proof]
  As the proof is very similar to the one of theorem~\ref{thm:S3axial}, we only sketch it. We have only to observe that setting $\bnu=\ee_{3}$, we have
  \begin{align*}
     & \bC=\Big[\bS\cdot \ee_{3} -3 \ee_{3} \odot \left(\ee_{3}\cdot \bS\cdot \ee_{3}\right)\Big]\times \ee_{3},
    \\
     & \bD=\bS\cdot \ee_{3} -4 \ee_{3} \odot \left(\ee_{3}\cdot \bS\cdot \ee_{3}\right) + 4 \ee_{3} \odot \ee_{3} \odot
    \big(\left(\ee_{3}\cdot \bS\cdot \ee_{3}\right)\cdot \ee_{3}\big),
    \\
     & \bE=
    \Big[\bS -5 \ee_{3} \odot (\bS\cdot \ee_{3})
      + \frac{20}{3} \ee_{3} \odot \ee_{3} \odot
      (\ee_{3}\cdot \bS\cdot \ee_{3})
      \Big]\times \ee_{3},
    \\
     & \bF=\big[\bS\cdot \ee_{3} -5 \ee_{3} \odot \left(\ee_{3}\cdot \bS\cdot \ee_{3}\right)
      + \frac{20}{3} \ee_{3} \odot \ee_{3} \odot
      \big(\left(\ee_{3}\cdot \bS\cdot \ee_{3}\right)\cdot \ee_{3}\big)
      \big]\times \ee_{3},
  \end{align*}
  where $\bC$, $\bD$, $\bE$ and $\bF$ are totally symmetric tensors, of respective orders three, four, five and six. The non-vanishing independent components of $\bC$ are given by
  \begin{align*}
    C_{111} & = S_{1123},
            &                          & C_{112}= \frac{1}{3} (2 S_{1223}-S_{1113}),
            &                          & C_{122}= \frac{1}{3} (S_{2223}-2 S_{1123}),
    \\
    C_{133} & = -\frac{1}{3} S_{2333},
            &                          & C_{222}= -S_{1223},
            &                          & C_{233}= \frac{1}{3} S_{1333}.
  \end{align*}
  The non-vanishing independent components of $\bD$ are given by
  \begin{align*}
     & D_{1111}= S_{11113},
     &                                    & D_{1112}= S_{11123},
     &                                    & D_{1122}= S_{11223},
     &                                    & D_{1133}= -\frac{1}{3} S_{11333},
     &                                    & D_{1222}= S_{12223},
    \\
     & D_{1233} = -\frac{1}{3} S_{12333},
     &                                    & D_{2222}= S_{22223},
     &                                    & D_{2233}= -\frac{1}{3} S_{22333},
     &                                    & D_{3333}= S_{33333}.
  \end{align*}
  The non-vanishing independent components of $\bE$ are given by
  \begin{align*}
     & E_{11111} = S_{11112},
     &                                                  & E_{11112}= \frac{1}{5} (4 S_{11122}-S_{11111}),
     &                                                  & E_{11122}= \frac{1}{5} (3 S_{11222}-2 S_{11112}),
    \\
     & E_{11133} = -\frac{1}{5} S_{11233},
     &                                                  & E_{11222}= \frac{1}{5} (2 S_{12222}-3 S_{11122}),
     &                                                  & E_{11233}= \frac{1}{15} (S_{11133}-2 S_{12233}),
    \\
     & E_{12222} = \frac{1}{5} (S_{22222}-4 S_{11222}),
     &                                                  & E_{12233}= \frac{1}{15} (2 S_{11233}-S_{22233}),
     &                                                  & E_{13333} = \frac{1}{5} S_{23333},
    \\
     & E_{22222}= -S_{12222},
     &                                                  & E_{22233}= \frac{1}{5} S_{12233},
     &                                                  & E_{23333}= -\frac{1}{5} S_{13333}.
  \end{align*}
  The non-vanishing independent components of $\bF$ are given by
  \begin{align*}
     & F_{11111}= S_{111123},
     &                                                   & F_{11112}= \frac{1}{5} (4 S_{111223}-S_{111113}),
     &                                                   & F_{11122}= \frac{1}{5} (3 S_{112223}-2 S_{111123}),
    \\
     & F_{11133}= -\frac{1}{5} S_{112333},
     &                                                   & F_{11222}= \frac{1}{5} (2 S_{122223}-3 S_{111223}),
     &                                                   & F_{11233}= \frac{1}{15} (S_{111333}-2 S_{122333}),
    \\
     & F_{12222}= \frac{1}{5} (S_{222223}-4 S_{112223}),
     &                                                   & F_{12233}= \frac{1}{15} (2 S_{112333}-S_{222333}),
     &                                                   & F_{13333}= \frac{1}{5} S_{233333},
    \\
     & F_{22222}= -S_{122223},
     &                                                   & F_{22233}= \frac{1}{5} S_{122333},
     &                                                   & F_{23333}= -\frac{1}{5} S_{133333}.
  \end{align*}
  The proof reduces then to check that if any of these tensors vanishes, then, the initial tensor $\bS$ has plane/axial symmetries, and conversely, which is a consequence of remark~\ref{rem:Sym_and_tensor_Coordinates} and~\eqref{eq:e3_instead_bnu}.
\end{proof}

%-----------------------------------------------------------------
\section{Plane symmetries of Elasticity tensors}
\label{sec:elasticity-tensors}
%-----------------------------------------------------------------

A linear elasticity tensor $\bE\in \Ela$ is defined as a fourth-order tensor having the major and the minor index symmetries,
$E_{ijkl}=E_{jikl}=E_{klij}$. Let $\bE^{s}\in \Sym^{4}(\RR^3)$ be its totally symmetric part and $\bA$ be its \emph{asymmetric} part (in the sense of Backus~\cite{Bac1970}). Their components write as follows
\begin{align*}
   & E^{s}_{ijkl}=\frac{1}{3}\left( E_{ijkl}+ E_{ikjl}+ E_{iljk}\right),
  \\
   & A_{ijkl}=\frac{1}{3}\left( 2 E_{ijkl}- E_{ikjl}- E_{iljk}\right).
\end{align*}
Set
\begin{equation}\label{eq:Def_Comp_a}
  \ba :=  2(\bd'-\bv')+ \frac{1}{6}(\tr \bd-\tr\bv)\, \id,
\end{equation}
where $\bd' = \bd-\frac{1}{3}(\tr \bd) \id$ and $\bv' = \bv-\frac{1}{3}(\tr \bv) \id$ are the deviatoric (\textit{i.e.} harmonic) parts of the dilatation and the Voigt tensors defined by \eqref{eq:dilatation_voigt}. We get then
\begin{equation*}
  \bA = \id \otimes_{2,2} \ba,
\end{equation*}
where $\otimes_{2,2}$ is the Young-symmetrized tensor product defined as in~\cite{OKDD2018a}, by
\begin{equation*}
  \by \otimes_{(2,2)}\!\bz = \frac{1}{3} \big( \by \otimes \bz +  \bz \otimes \by
  - \by \otimesbar \bz - \bz \otimesbar \by \big),\quad (\by \otimesbar \bz)_{ijkl} : = \frac{1}{2} (y_{ik}z_{jl} + y_{il}z_{jk}).
\end{equation*}
We can write then $\bE \simeq (\bE^s, \ba)$, where $\bE^s\in \Sym^{4}(\RR^3)$ and $\ba \in \Sym^2(\RR^3)$ are totally symmetric tensors. Moreover, this decomposition, related to the harmonic decomposition of $\Ela$, is equivariant, which means that
\begin{equation*}
  g\star \bE \simeq (g\star\bE^{s}, g\star \ba)\quad \forall g \in \OO(3).
\end{equation*}

\begin{thm}\label{thm:elasticity}
  Let $\bE$ be an elasticity tensor, $\bE^{s}$ be its totally symmetric part and $\bd=\tr_{12}\bE$ be the dilatation tensor. Necessary and sufficient conditions for $\bnu$ to be a normal to a symmetry plane --- or equivalently to be an axis of symmetry --- for $\bE$ are
  \begin{equation}\label{eq:CondEla}
    \begin{cases}
      \Big[\bE^{s}\cdot \bnu -3 \bnu \odot \left(\bnu\cdot \bE^s\cdot \bnu\right)\Big]\times \bnu=0,
      \\
      \left(\bd\cdot \bnu\right)\times \bnu=0.
    \end{cases}
  \end{equation}
\end{thm}

\begin{proof}
  From the equivariant decomposition $\bE \simeq (\bE^{s}, \ba)$, we get
  \begin{equation*}
    \bs(\bnu)\star \bE \simeq (\bs(\bnu)\star\bE^{s}, \bs(\bnu)\star \ba), \quad \forall g \in \OO(3).
  \end{equation*}
  Hence, $\bs(\bnu)$ is a plane symmetry of $\bE$ if and only if $\bs(\bnu)$ is a plane symmetry of both $\bE^{s}$ and $\ba$. By proposition \ref{prop:Axial_Second_Order} and theorem \ref{thm:S4toS6}, this recasts as
  \begin{equation}\label{eq:CondEla_a}
    \begin{cases}
      \Big[\bE^{s}\cdot \bnu -3 \bnu \odot \left(\bnu\cdot \bE^s\cdot \bnu\right)\Big]\times \bnu=0,
      \\
      \left(\ba\cdot \bnu\right)\times \bnu =0.
    \end{cases}
  \end{equation}
  Suppose now that $\bs(\bnu)$ is the normal to a common plane symmetry of $\bE^{s}$ and $\ba$, so that $\bs(\bnu)$ is also a plane symmetry of
  \begin{equation}\label{eq:Trace_Es}
    \tr \bE^{s}=\frac{1}{3}( \bd+2 \bv)= \bd - \frac{1}{3}\ba.
  \end{equation}
  Therefore, by proposition~\ref{prop:Axial_Second_Order}, we get
  \begin{equation*}
    \left(\Big(\bd - \frac{1}{3}\ba\Big)\cdot \bnu\right)\times \bnu= \left(\bd\cdot \bnu\right)\times \bnu-\frac{1}{3} \left(\ba\cdot \bnu\right)\times \bnu=0 \Rightarrow (\bd\cdot \bnu)\times \bnu=0.
  \end{equation*}
  Conversely, if equations~\eqref{eq:CondEla} holds, we deduce from theorem~\ref{thm:S4toS6} that $\bs(\bnu)$ is a plane symmetry of $\bE^{s}$, and thus a plane symmetry of $\tr \bE^{s}$. Using proposition~\ref{prop:Axial_Second_Order} applied to $\tr \bE^{s}$ and by~\eqref{eq:Trace_Es}, we deduce that $(\ba\cdot \bnu)\times \bnu=0$. Hence, $\bs(\bnu)$ is also a plane symmetry of $\ba$, which ends the proof.
\end{proof}

The first condition in theorem \ref{thm:elasticity} is polynomial of degree $n=4$ in $\bnu$ (instead of degree $2n=8$ for the genuine algebraic condition $\bs(\bnu)\star \bE=\bE$). It is ranked in a totally symmetric third-order tensor (instead of fourth-order tensor for the genuine algebraic condition). The second condition is polynomial of degree two in $\bnu$ and is ranked in a vector.

Theorem \ref{thm:elasticity} determines the normals to all the plane symmetries of an elasticity tensor in an arbitrary orthonormal basis (and this for any symmetry class and/or any particular relationship between material parameters). It is a real improvement compared to Cowin--Mehrabadi theorems, since it does not involve anymore unit vectors $\btau$ perpendicular to $\bnu$. For instance in the particular case $\bd'=\bv'=\ba'=0$ of an elasticity tensor with spherical dilatation and Voigt's tensors, the solutions of the first equation in theorem~\ref{thm:elasticity} are the normals to all the plane symmetries of $\bE$, whereas the solutions of the first condition $\left[ \left(\bnu\cdot \bE\cdot \bnu\right)\cdot\bnu\right] \times \bnu=0$ in Cowin-Mehrabadi's theorem, the condition on $\bnu$ only, gives extra vectors $\bnu$ that are not necessary normals to plane symmetries of $\bE$.

%-----------------------------------------------------------------
\section{Plane and axial symmetries of Piezo-electricity}
\label{sec:piezoelectricity-tensors}
%-----------------------------------------------------------------

Consider now the linear piezo-electricity constitutive equations
\begin{equation}\label{eq:EPeps}
  \begin{cases}
    \bepsilon= \bE^{-1}:\bsigma+  \vec E \cdot \bP
    \\
    \vec D=\bP:\bsigma+ \pmb\varepsilon_0^\sigma\cdot \vec E
  \end{cases}
  ,
  \qquad
  \begin{cases}
    \epsilon_{ij}= (\bE^{-1})_{ijkl}\sigma_{kl}+ P_{kij} E_k
    \\
    D_i=P_{ikl}\sigma_{kl}+ \varepsilon_{0\, ik}^\sigma E_k
  \end{cases}
\end{equation}
between the electric field $\vec E$, the (symmetric) second-order stress tensor $\bsigma$, the electric displacement $\vec D$ and the (symmetric) second-order strain tensor $\bepsilon$,
where
\begin{itemize}
  \item $\bE^{-1}\in \Ela$ is the compliance fourth-order tensor (the inverse of the elasticity tensor, with the same index symmetries),
  \item $\bP\in \Piez$ is the piezo-electricity third-order tensor,
  \item $\pmb\varepsilon_0^\sigma\in \Perm$ is the (symmetric) second-order dielectric permittivity tensor.
\end{itemize}
Each piezo-electricity tensor $\bP$ has the index symmetry $P_{ijk}=P_{ikj}$, \textit{i.e.} it is such as $\bP^{(23)}= \bP$, where the notation $\bT^{(23)}$ means here the
symmetrized of a third-order tensor $\bT$ with respect to the second and third
subscripts,
\begin{equation}\label{eq:tT}
  (\bT^{(23)})_{ijk}:= \frac{1}{2}(T_{ijk}+T_{ikj})
\end{equation}

One can study independently the plane/axial symmetries of each constitutive tensor $\bE$ (fourth-order), $\bP$ (third-order) and $\pmb\varepsilon_0^\sigma$ (second-order) or may look for the plane/axial symmetries of the whole set $(\bE,\bP,\pmb\varepsilon_0^\sigma)$. The latest problem comes down to finding unit vectors $\bnu$ such that
\begin{equation*}
  (\bs(\bnu)\star \bE,\bs(\bnu)\star \bP,\bs(\bnu)\star \pmb\varepsilon_0^\sigma)=(\bE,\bP,\pmb\varepsilon_0^\sigma)
  \quad \text{for plane symmetries},
\end{equation*}
or
\begin{equation*}
  (\br(\bnu,\pi)\star \bE,\br(\bnu,\pi)\star \bP,\br(\bnu,\pi)\star \pmb\varepsilon_0^\sigma)=(\bE,\bP,\pmb\varepsilon_0^\sigma)
  \quad \text{for axial symmetries}.
\end{equation*}

As in the case of elasticity tensors, we will use here some equivariant decomposition. Recall first the explicit equivariant decomposition of $\bE$
\begin{equation*}
  \bE = \bE^{s}+\id \otimes_{2,2} \ba,
\end{equation*}
where $\bE^{s}$ is a totally symmetric fourth-order tensor and $\ba$ is a symmetric second-order tensor (see~\autoref{sec:elasticity-tensors}). We have a similar decomposition for $\bP$ (related to its harmonic decomposition, see \cite{GW2002})
\begin{equation}\label{eq:EquivP}
  \bP\in \Piez\mapsto (\bP^{s},\bh,\ww)
  \qquad
  \begin{cases}
    \ww := \frac{3}{4} \left(\tr_{23}( \bP)-\tr (\bP^{s})\right)\in \HH^{1}(\RR^{3}),
    \\
    \bh := \frac{2}{3}\left(\lc:\bP\right)^{s}\in \HH^{2}(\RR^{3}),
  \end{cases}
\end{equation}
where $\lc$ is the third-order Levi-Civita symbol (with components $\varepsilon_{ijk}=\det(\ee_{i}, \ee_{j}, \ee_{k})$ in any direct orthonormal basis $(\ee_{i})$), not to be confused with dielectric permittivity tensor $\pmb \varepsilon_0^\sigma\in \Sym^2(\RR^3)$, nor with the strain tensor $\bepsilon\in \Sym^2(\RR^3)$, so that $(\lc:\mathbf{P})_{ij}=\varepsilon_{ikl}P_{klj}$. Here, $\HH^{k}(\RR^3)$ stands for the vector space of harmonic tensors of order $k$, harmonic meaning totally symmetric and traceless. In this decomposition, the component $\bP^{s}\in \Sym^{3}(\RR^3) $ is the totally symmetric part of $\bP$
\begin{equation*}
  (\bP^{s})_{ijk} = \frac{1}{3}(P_{ijk}+P_{jik}+P_{kij}).
\end{equation*}
As detailed in~\cite{Spe1970,JCB1978,GW2002}, the explicit decomposition of $\bP$ writes as
\begin{equation*}
  \bP = \bP^{s}+  \ww\otimes\id- \ww \odot\id+(\lc\cdot\bh)^{(23)},
\end{equation*}
where
$(\lc\cdot \bh)^{(23)}_{ijk}
  =\frac{1}{2}((\lc\cdot \bh)_{ijk}+(\lc\cdot \bh)_{ikj})=\frac{1}{2}(\varepsilon _{ijl} h_{lk}+\varepsilon _{ikl} h_{lj})
$ (see~\eqref{eq:tT}). This decomposition~\eqref{eq:EquivP} is equivariant relatively to the action of the orthogonal group given by
\begin{equation*}
  g \star \bP\simeq ( g\star \bP^{s}, g \,\hat \star\, \bh, g\star \ww)\; \qquad g\in \OO(3),
\end{equation*}
where $g \,\hat \star\, \bh:=(\det g) g \star \bh$. In other words, $\bP^{s}$ is a third-order tensor, $\ww$ is a vector but $\bh$ is a second-order \emph{pseudo-tensor} (this is due to the contraction with Levi-Civita tensor in its definition). In these formulas, $\hat \star=\det () \, \star$ is so-called \emph{twisted} action, the one to be considered for pseudo-tensors of any order.

Finally one can recapitulate all these decompositions into the following decomposition of the triplet $(\bE, \bP, \pmb\varepsilon_0^\sigma)\in \Ela \oplus \Piez \oplus \Perm$
\begin{equation*}
  (\bE, \bP, \pmb\varepsilon_0^\sigma)\simeq (\bE^{s}, \bP^{s}, \pmb\varepsilon_0^\sigma, \ba, \bh, \ww) ,
\end{equation*}
\textit{i.e.} as a decomposition into a totally symmetric tensors $\bE^{s}$ (of order 4), $\bP^{s}$ (of order 3), $\pmb\varepsilon_0^\sigma$ and $\ba$ (of order 2), an harmonic pseudo-tensor $\bh$ (of order 2) and a vector $\ww$,
such that
\begin{equation*}
  g \star (\bS, \bP^{s}, \pmb\varepsilon_0^\sigma, \ba, \bh, \ww) =
  ( g\star\bS,  g\star\bP^{s},  g\star\pmb\varepsilon_0^\sigma, g\star\ba, g \,\hat \star \, \bh,  g\star\ww)
  \;\forall g\in \OO(3).
\end{equation*}

Following the same proof as for theorem~\ref{thm:elasticity} with the use of reduced equations for totally symmetric tensor
(and using point (2) of proposition~\ref{prop:Axial_Second_Order} for second-order pseudo-tensor $\bh$), we obtain reduced equations for the existence of plane/axial symmetries for the linear piezo-electricity constitutive equations \eqref{eq:EPeps}.

\begin{thm}\label{thm:piezo-electricity}
  Let $(\bE, \bP,\pmb\varepsilon_0^\sigma) \in \Ela \oplus \Piez \oplus \Perm$ be the triplet of elasticity, piezo-electricity and dielectric permittivity tensors, where $\bE^{s}$ and $\bP^{s}$ are the totally symmetric part of $\bE$ and $\bP$,
  \begin{equation*}
    \bd=\tr_{12}\bE, \qquad \bh = \frac{2}{3} \left(\lc:\bP\right)^{s},
    \qquad \ww = \frac{3}{4} \left(\tr_{23} (\bP)-\tr (\bP^{s})\right),
  \end{equation*}
  and $\lc$ is the third-order Levi-Civita symbol. Let $\bnu$ be a unit vector, then:
  \begin{enumerate}
    \item $\bs(\bnu)$ is a plane symmetry of $(\bE, \bP,\pmb\varepsilon_0^\sigma)$ if and only if
          \begin{equation}\label{eq:Planar_Piezo}
            \begin{cases}
              \left[\bE^{s} \cdot \bnu
                -3 \bnu\odot \left(\bnu\cdot\bE^{s} \cdot \bnu\right) \right]\times \bnu= 0,
              \\
              \bP^{s} \cdot \bnu
              -2 \bnu\odot \left(\bnu\cdot\bP^{s} \cdot \bnu\right) = 0,
              \\
              (\pmb\varepsilon_0^\sigma\cdot \bnu)\times \bnu=0,
              \\
              (\bd\cdot \bnu)\times \bnu=0,
              \\
              \bh-2\bnu\odot (\bh\cdot \bnu) = 0, \\
              \ww\cdot \bnu = 0.
            \end{cases}
          \end{equation}
    \item $\br(\bnu,\pi)$ is an axis symmetry axis of $(\bE, \bP,\pmb\varepsilon_0^\sigma)$ if and only if
          \begin{equation}\label{eq:Axis_Piezo}
            \begin{cases}
              \left[\bE^{s} \cdot \bnu
                -3 \bnu\odot \left(\bnu\cdot\bE^{s} \cdot \bnu\right) \right]\times \bnu= 0,
              \\
              \left[\bP^{s} -3 \bnu \odot (\bP^{s}\cdot \bnu)\right]\times \bnu = 0,
              \\
              (\pmb\varepsilon_0^\sigma\cdot \bnu)\times \bnu=0,
              \\
              (\bd\cdot \bnu)\times \bnu=0,
              \\
              (\bh\cdot\bnu) \times \bnu = 0,
              \\
              \ww\times  \bnu = 0.
            \end{cases}
          \end{equation}
  \end{enumerate}
\end{thm}

Theorem~\ref{thm:piezo-electricity} is constructive, with no equation involving vectors $\btau$ perpendicular to the normal/axis $\bnu$. The sets of equations~\eqref{eq:Planar_Piezo} and~\eqref{eq:Axis_Piezo} determine all plane/axial symmetries of a triplet $(\bE, \bP,\pmb\varepsilon_0^\sigma)$, independently of their symmetry class. Compared to the standard action of an orthogonal transformation, the degree of the polynomials in $\bnu$ are again reduced by a factor 2.

% ----------------------------------------------------------------
\section{Conclusion}
% ----------------------------------------------------------------

We have established necessary and sufficient conditions for a unit vector $\bnu$ to be a normal to a symmetry plane or the axis of a rotational symmetry of a constitutive tensor for main problems in continuum mechanics. These conditions are \emph{reduced} since they are of degree $n$ in $\bnu$, rather than $2n$ as it is the case for the genuine equations. They are coordinate free and constructive: one can use them to compute explicitly the vectors $\bnu$ by solving some algebraic equations. These conditions have been summarized as theorems~\ref{thm:S3plane}, \ref{thm:S3axial} and \ref{thm:S4toS6} for totally symmetric tensors of order three up to six. They have been extended to the elasticity tensor, as theorem \ref{thm:elasticity}, and to the linear piezo-electricity, as theorem~\ref{thm:piezo-electricity}, using the harmonic decomposition.

Finally, these equations are related to the attempt to \emph{generalize the concept of eigenvectors} to higher order tensors~\cite{Qi2007,CS2013} and we will conclude this paper by a discussion on this topic. For a second-order symmetric tensor $\ba$, a vector $\bnu$ is an eigenvector of $\ba$ if and only if $(\ba\cdot \bnu)\times \bnu=0$. In that case, it defines clearly an axis of symmetry of $\ba$. For a tensor $\bS$ of order $n \ge 3$, Qi, in~\cite{Qi2007}, has defined a unit vector $\bnu$ to be a \emph{Z-eigenvector} if it is a solution of
\begin{equation*}
  \bS\rcont{n-1}\bnu^{n-1}=\lambda \bnu \iff \left( \bS\rcont{n-1}\bnu^{n-1}\right)\times \bnu=0.
\end{equation*}
By point (2) of lemma~\ref{lem:Sym_With_Contraction}, each unit vector $\bnu$ defining an axial symmetry of $\bS$ is necessary a $Z$-eigenvector but the converse is false in general. Therefore, one could weaken Qi's definition of a \emph{generalized eigenvector} and simply define it a generator a symmetry axis of $\bS$. We will call such a $\bnu$ as an \emph{A-eigenvector} of $\bS$.

The \emph{characteristic equations} for a vector $\xx$ to be an \emph{A-eigenvector} of $\bS$, are thus deduced from present paper because our equations are homogeneous. For instance for totally symmetric tensors of order $n=3$ up to order $6$ (setting $\xx=\norm{\xx} \bnu$ in Theorems~\ref{thm:S3axial} and \ref{thm:S4toS6} ), we get
\begin{align*}
   & \textrm{$\bS$ of order 3:} &  &
  \Big[\norm{\xx}^2 \bS -3 \xx \odot (\bS\cdot \xx)\Big]\times \xx=0,
  \\
   & \textrm{$\bS$ of order 4:} &  & \Big[\norm{\xx}^2\, \bS\cdot \xx -3 \xx \odot \left(\xx\cdot \bS\cdot \xx\right)\Big]\times \xx=0,
  \\
   & \textrm{$\bS$ of order 5:} &  & \Big[\|\xx \|^4 \,\bS -5\|\xx \|^2\, \xx \odot (\bS\cdot \xx)
    + \frac{20}{3} \xx \odot \xx \odot
    (\xx\cdot \bS\cdot \xx)
    \Big]\times \xx = 0,
  \\
   & \textrm{$\bS$ of order 6:} &  & \big[\norm{\xx}^4\, \bS\cdot \xx  -5 \norm{\xx}^2 \, \xx \odot \left(\xx\cdot \bS\cdot \xx\right)
    + \frac{20}{3} \xx \odot \xx \odot \big(\left(\xx\cdot \bS\cdot \xx\right)\cdot \xx\big) \big]\times \xx=0.
\end{align*}
Similar results for totally symmetric tensors of any order are available in~\cite{ODKD2019}.

\section{Acknowledgments*}

The authors warmly thank the referees for very helpful comments on generalized eigenvectors and for communicating us significative additional references.

%-----------------------------------------------------------------


\begin{thebibliography}{10}

\bibitem{Bac1970}
G.~Backus.
\newblock {A} geometrical picture of anisotropic elastic tensors.
\newblock {\em Rev. Geophys.}, 8(3):633--671, 1970.

\bibitem{CS2013}
D.~Cartwright and B.~Sturmfels.
\newblock The number of eigenvalues of a tensor.
\newblock {\em Linear Algebra and its Applications}, 438(2):942--952, Jan.
  2013.

\bibitem{CVC2001}
P.~Chadwick, M.~Vianello, and S.~Cowin.
\newblock A new proof that the number of linear elastic symmetries is eight.
\newblock {\em J. Mech. Phys. Solids}, 49:2471--2492, 2001.

\bibitem{Cow1989}
S.~Cowin.
\newblock Properties of the anisotropic elasticity tensor.
\newblock {\em Q. J. Mech. Appl. Math.}, 42:249--266, 1989.

\bibitem{CM1987}
S.~Cowin and M.~Mehrabadi.
\newblock On the identification of material symmetry for anisotropic elastic
  materials.
\newblock {\em Q. J. Mech. Appl. Math.}, 40:451--476, 1987.

\bibitem{DADKO2019}
R.~Desmorat, N.~Auffray, B.~Desmorat, B.~Kolev, and M.~Olive.
\newblock Generic separating sets for three-dimensional elasticity tensors.
\newblock {\em Proceedings of the Royal Society A}, 475(2226):20190056, 2019.

\bibitem{Fed1968}
F.~I. Fedorov.
\newblock {\em Theory of Elastic Waves in Crystals}.
\newblock Plenum Press, New York, 1968.

\bibitem{FV1996}
S.~Forte and M.~Vianello.
\newblock {S}ymmetry classes for elasticity tensors.
\newblock {\em J. Elasticity}, 43(2):81--108, 1996.

\bibitem{FGB1998}
M.~Francois, Y.~Berthaud, and G.~Geymonat.
\newblock Determination of the symmetries of an experimentally determined
  stiffness tensor: application to acoustic measurements.
\newblock {\em Int. J. Solids Structures}, 35:4091--4106, 1998.

\bibitem{GW2002}
G.~Geymonat and T.~Weller.
\newblock Symmetry classes of piezoelectric solids.
\newblock {\em Comptes rendus de l'Acad\'emie des Sciences. S\'erie I},
  335:847--852, 2002.

\bibitem{Gor2017}
A.~Gorodentsev.
\newblock {\em Algebra II: Textbook for Students of Mathematics}.
\newblock Springer, 2017.

\bibitem{Jar1994}
J.~Jaric.
\newblock On the condition for the existence of a plane of symmetry for
  anisotropic material.
\newblock {\em Mech. Res. Comm.}, 21(2):153--174, 1994.

\bibitem{JCB1978}
J.~Jerphagnon, D.~Chemla, and R.~Bonneville.
\newblock The description of the physical properties of condensed matter using
  irreducible tensors.
\newblock {\em Advances in Physics}, 27(4):609--650, 1978.

\bibitem{Kol1966}
I.~Kolodner.
\newblock Existence of longitudinal waves in anisotropic media.
\newblock {\em J. Acous. Soc. Amer.}, 40(3):730--731, 1966.

\bibitem{Nor1989}
A.~N. Norris.
\newblock On the acoustic determination of the elastic moduli of anisotropic
  solids and acoustic conditions for the existence of symmetry planes.
\newblock {\em The Quarterly Journal of Mechanics and Applied Mathematics},
  42(3):413--426, 1989.

\bibitem{ODKD2019}
M.~Olive, B.~Desmorat, B.~Kolev, and R.~Desmorat.
\newblock Reduced algebraic conditions for plane/axial tensorial symmetries.
\newblock {\em hal-02291487}, 2019.

\bibitem{OKDD2018b}
M.~Olive, B.~Kolev, R.~Desmorat, and B.~Desmorat.
\newblock Characterization of the symmetry class of an elasticity tensor using
  polynomial covariants.
\newblock 2018.

\bibitem{OKDD2018a}
M.~Olive, B.~Kolev, R.~Desmorat, and B.~Desmorat.
\newblock Harmonic factorization and reconstruction of the elasticity tensor.
\newblock {\em Journal of Elasticity}, 101:132--67, 2018.

\bibitem{Qi2007}
L.~Qi.
\newblock Eigenvalues and invariants of tensors.
\newblock {\em J. Math. Anal. Appl.}, 325:1363--1377, 2007.

\bibitem{Spe1970}
A.~Spencer.
\newblock A note on the decomposition of tensors into traceless symmetric
  tensors.
\newblock {\em Int. J. Engng Sci.}, 8:475--481, 1970.

\bibitem{Sti1965}
M.~Stippes.
\newblock Steady state waves in anisotropic media.
\newblock {\em Ann. Meeting Soc. Eng. Sci.}, unpublished, 1965.

\bibitem{Tin2003}
T.~Ting.
\newblock Generalized {C}owin-{M}ehrabadi theorems and a direct proof that the
  number of linear elastic symmetries is eight.
\newblock {\em International Journal of Solids and Structures}, 40:7129--7142,
  2003.

\end{thebibliography}
\end{document}